\documentclass[11pt,a4paper,reqno]{amsart}
\usepackage{amsmath}
\usepackage[english]{babel}
\usepackage{latexsym}
\usepackage{amssymb}
\usepackage{amscd}
\usepackage{amsgen,amstext,amsbsy,amsopn}
\usepackage{math rsfs}
\usepackage{bm,bbm}
\usepackage{amsthm,epsfig,graphicx,graphics}
\usepackage[latin1]{inputenc}
\usepackage{xspace}
\usepackage{amsxtra}
\usepackage{color}
\usepackage{dsfont}

\usepackage[pdftex]{hyperref}

\usepackage{geometry}
\newcommand{\alj}{\alpha_n}

\newcommand{\kj}{k_n}

\newcommand{\tj}{\bar{t}_{n,\eps}}
\newcommand{\tone}{\bar{t}_{1,\eps}}

\numberwithin{equation}{section}
\newcommand{\bdm}{\begin{displaymath}}
\newcommand{\edm}{\end{displaymath}}
\newcommand{\bdn}{\begin{eqnarray}}
\newcommand{\edn}{\end{eqnarray}}
\newcommand{\bay}{\begin{array}{c}}
\newcommand{\eay}{\end{array}}
\newcommand{\ben}{\begin{enumerate}}
\newcommand{\een}{\end{enumerate}}
\newcommand{\beq}{\begin{equation}}
\newcommand{\eeq}{\end{equation}}
\newcommand{\beqn}{\begin{eqnarray}}
\newcommand{\eeqn}{\end{eqnarray}}
\newcommand{\bml}[1]{\begin{multline} #1 \end{multline}}
\newcommand{\bmln}[1]{\begin{multline*} #1 \end{multline*}}

\newcommand{\lf}{\left}
\newcommand{\ri}{\right}

\newcommand{\rv}{\mathbf{r}}

\newcommand{\aae}{a_{\eps}}

\newcommand{\deps}{\delta_{\eps}}

\newcommand{\de}{d_{\eps}}
\newcommand{\diff}{\mathrm{d}}
\newcommand{\eps}{\varepsilon}

\newcommand{\dist}{\mathrm{dist}}

\newcommand{\nuv}{\bm{\nu}}

\newcommand{\glm}{\Psi^{\mathrm{GL}}}

\newcommand{\gldom}{\mathscr{D}^{\mathrm{GL}}}
\newcommand{\aav}{\mathbf{A}}
\newcommand{\aavm}{\mathbf{A}^{\mathrm{GL}}}

\newcommand{\hex}{b}
\newcommand{\theo}{\Theta_0}

\newcommand{\glfe}{\mathcal{G}_{\eps}^{\mathrm{GL}}}
\newcommand{\gled}{e_{\eps}^{\mathrm{GL}}}
\newcommand{\glee}{E_{\eps}^{\mathrm{GL}}}

\newcommand{\annf}{\mathcal{G}_{\ann}}

\newcommand{\curv}{k(s)}

\newcommand{\eones}{E^{\mathrm{1D}}_{\star}}

\newcommand{\pot}{V_{k,\alpha}}

\newcommand{\curl}{\mbox{curl}}

\newcommand{\ann}{\mathcal{A}_{\eps}}
\newcommand{\annt}{\tilde{\mathcal{A}}_{\eps}}

\newcommand{\cell}{\mathcal{C}}

\newcommand{\cellj}{\mathcal{C}_n}

\newcommand{\half}{\mbox{$\frac{1}{2}$}}

\newcommand{\tx}{\textstyle}

\newcommand{\neps}{N_{\eps}}

\newcommand{\R}{\mathbb{R}}
\newcommand{\N}{\mathbb{N}}

\newcommand{\C}{\mathbb{C}}

\newcommand{\E}{\mathcal{E}}
\newcommand{\A}{\mathcal{A}}

\newcommand{\LL}{\mathcal{L}}

\newcommand{\Om}{\Omega}

\newcommand{\dd}{\partial}

\newcommand{\ab}{\A_{>}}
\newcommand{\abt}{\tilde{\A}_{>}}

\newcommand{\supp}{\mathrm{supp}}

\newcommand{\Hcc}{H_{\mathrm{c}2}}
\newcommand{\Hccc}{H_{\mathrm{c}3}}

\newcommand{\logi}{|\log \eps| ^{\infty}}

\newtheorem{teo}{Theorem}[section]
\newtheorem{lem}{Lemma}[section]
\newtheorem{pro}{Proposition}[section]

\theoremstyle{remark}
\newtheorem{remark}{Remark}[section]

\newcommand{\fk}{f_{k}}

\newcommand{\fO}{f_0}

\newcommand{\fone}{\E^{\mathrm{1D}}}

\newcommand{\foneO}{\E ^{\rm 1D}_{0}}

\newcommand{\fc}{\E ^{\rm{corr}}}

\newcommand{\eone}{E^{\mathrm{1D}}}

\newcommand{\eoneo}{E ^{\rm 1D}_{0}}

\newcommand{\alk}{\alpha(k)}

\newcommand{\alO}{\alpha_0}

\begin{document}


\title{Effects of  boundary curvature  on surface superconductivity}

\author[M. Correggi]{Michele CORREGGI}
\address{Dipartimento di Matematica e Fisica, Universit\`{a} degli Studi Roma Tre, L.go San Leonardo Murialdo, 1, 00146, Rome, Italy.}
\email{}

\author[N. Rougerie]{Nicolas ROUGERIE}
\address{CNRS \& Universit\'e Grenoble Alpes, LPMMC (UMR 5493), B.P. 166, F-38042 Grenoble, France}
\email{nicolas.rougerie@lpmmc.cnrs.fr}

\date{November, 2015}

%

\begin{abstract} 
We investigate, within 2D Ginzburg-Landau theory, the ground state of a type-II superconducting cylinder in a parallel magnetic field varying between the second and third critical values. In this regime, superconductivity is restricted to a thin shell along the boundary of the sample and is to leading order constant in the direction tangential to the boundary. We exhibit a correction to this effect, showing that the curvature of the sample affects the  distribution of superconductivity. 
\end{abstract}

\maketitle

\tableofcontents

\section{Introduction and Main Result}\label{sec:intro}

The response of type-II superconductors to external magnetic fields is a rich source of fascinating mathematical problems~\cite{BBH2,FH-book,SS2,Sig}. Physically, this is because of the occurrence of mixed states where normal and superconducting regions may coexist. If the sample is a very long cylinder and the applied magnetic field is parallel to it, one may adopt a 2D description on a cross-section of the cylinder. One then distinguishes mainly two types of mixed phases:
\begin{itemize}
 \item The vortex lattice (Abrikosov lattice~\cite{Abr}) where the normal regions take the form of vortices, and are arranged on a triangular lattice embedded in a sea of superconducting material;
 \item The surface superconductivity state where the whole bulk of the sample is in the normal state, and superconductivity only survives close to the boundary.
\end{itemize}
The second situation shall concern us here. In this regime, the magnetic field is very large, varies between two critical values $\Hcc$ and $\Hccc$, and mostly penetrates the sample. That superconducting electrons may still exist because of boundary effects is a highly non-trivial fact, first derived in~\cite{SJdG}. At leading order, the phenomenon may be completely understood by considering the case of an infinite half-plane sample, with a straight boundary. It thus has some universal features: although superconducting electrons concentrate along the boundary, the geometry of the latter does not affect their distribution much. This is because the density of superconducting electrons essentially only varies in the direction normal to the boundary, as we proved rigorously in~\cite{CR1,CR2}, following several earlier contributions~\cite{Alm,AH,FHP,FH1,FH2,FK,LP,Pan} (see~\cite{FH-book} for a review).

However, in order to prove some of the most refined results in~\cite{CR1,CR2}, we had to precisely estimate subleading order contributions to the energy. These do not share the universal character of the leading order, in that they depend on the sample, via the curvature of its boundary. This is reminiscent of earlier works, e.g., \cite{FH1} (see~\cite[Chapters~13 and~15]{FH-book}  for extensive references), where it has been proved that, when decreasing the magnetic field just below $\Hccc$, superconductivity appears first where the curvature of the boundary is maximum. In this paper we aim at evaluating the effect of sample curvature in the whole regime of magnetic fields comprised between $\Hcc$ and $\Hccc$. We shall give a simple expression of the curvature dependent contribution to surface superconductivity, which, again, appears only at subleading order in the energy, and thus requires rather refined estimates.  

\medskip

Our setting is the following: we consider the Ginzburg-Landau functional (in convenient units whose relation to other conventions is discussed\footnote{Note that in~\cite{CR1,CR2}, an extra factor of $b$ has been mistakenly inserted in front of the last term. This does not have any incidence on the results since this term is negligible in the regime of our interest.} in~\cite{CR1,CR2})
\beq\label{eq:GL func eps}
	\glfe[\Psi,\aav] = \int_{\Om} \diff \rv \: \bigg\{ \bigg| \bigg( \nabla + i \frac{\aav}{\eps^2} \bigg) \Psi \bigg|^2 - \frac{1}{2 \hex \eps^2} \lf( 2|\Psi|^2 - |\Psi|^4 \ri) + \frac{1}{\eps^4} \lf| \curl \aav - 1 \ri|^2 \bigg\}.
\eeq
The domain $\Omega \subset \R ^2$ represents the cross-section of an infinitely long cylinder of superconducting material. We assume that it is bounded, simply connected and that its boundary is smooth. The applied magnetic field is perpendicular to $\Omega$. 

We shall denote
\beq
	\glee : = \min_{(\Psi, \aav) \in \gldom} \glfe[\Psi,\aav],
\eeq
with
\beq
	\gldom : = \lf\{ (\Psi,\aav) \in H^1(\Om;\C) \times H^1(\Om;\R^2) \ri\},
\eeq
and denote by $ (\glm,\aavm) $ a minimizing pair. We recall that $|\glm| ^2$ gives the local relative density of superconducting electrons (bound in Cooper pairs) and that $\curl \: \aavm$ is the induced magnetic field in the sample.

We are interested in the behavior of $|\glm|$ in the surface superconductivity regime
\beq
	\label{eq:b condition}
	1 < b < \theo^{-1},
\eeq
where $\theo$ is the minimal ground state energy of the shifted harmonic oscillator on the half-line:
\begin{equation}\label{eq:def theo}
 \theo := \min_{\alpha \in \R } \min_{  \lf\| u \ri\|_2 = 1 } \int_{0} ^{+ \infty} \diff t \lf\{ |\dd_t u | ^2 + (t+\alpha) ^2 |u| ^2 \ri\}. 
\end{equation}
This corresponds to asking that the applied magnetic field varies between $\Hcc$ and $\Hccc$, and we shall also assume that $\eps$ is a small parameter, in order to prove asymptotic results in the limit $\eps\to 0$. Physically this means we consider an ``extreme''  type-II superconductor.

In the parameter regime of our interest, the GL order parameter is concentrated near the boundary of the sample and the induced magnetic field is very close to the (constant) applied one. To leading order, superconducting electron pairs are uniformly distributed along the boundary as a function of the tangential variable. The main variations are in the direction normal to the boundary. Our main result in this note is an asymptotic estimate for $|\glm| ^4$ which exhibits a subleading (in $\eps$) curvature-dependent correction. This shows that variations of the boundary's curvature influence the superconductivity distribution.

In our previous papers \cite{CR1,CR2}, we have emphasized the role played by the simplified functional
\begin{equation}
\label{eq:1D func}
\fone_{k,\alpha}[f] : = \int_0^{c_0|\log\eps|} \diff t (1-\eps k t )\lf\{ \lf| \partial_t f \ri|^2 + \pot(t) f^2 - \tx\frac{1}{2b} \lf(2 f^2 - f^4 \ri) \ri\}
\end{equation}
where
\beq
	\label{eq:pot}
	\pot(t) : = \frac{(t + \alpha - \frac12 \eps k t ^2 )^2}{(1-\eps k t ) ^2}
\eeq 
and $c_0$ is a (somewhat arbitrarily) fixed, large enough, constant. Note that in the limit $ \eps \to 0  $ the above expression reduces to a 1D nonlinear energy independent of the curvature, which is known to provide the leading order contribution to the GL asymptotics (see \cite{CR1} and references therein).

This functional should be  thought of as giving the GL energy to subleading order in the case of a sample $\Om$ which is a ``disc'' of curvature $k$, i.e., either a disc of radius $R = k ^{-1}$ when $k > 0$ or the exterior of such a disc when $k < 0$. The variable $t$ corresponds to the coordinate normal to the boundary (in units of $\eps$). Denoting by $s$ the tangential coordinate, $f(t) e ^{-i\alpha s}$ gives an ansatz for the GL order parameter in boundary coordinates. To get an optimal energy, this functional should be minimized with respect to both the function $f$ and the number $\alpha$, leading to an optimal profile $\fk$, an optimal phase $\alk$ and an optimal energy $\eone_\star \left(k\right)$.

Since we work in the regime $\eps \to 0$, it is natural to try to consider a perturbative expansion of $\eone_\star \left(k(s)\right)$. We then get that the leading order is given by the $k=0$ functional (corresponding to a half-plane sample and extensively studied in the literature~\cite{AH,FHP,Pan}):
\beq\label{eq:1D func bis}
\fone_{0,\alpha}[f] : = \int_0^{+\infty} \diff t \lf\{ \lf| \partial_t f \ri|^2 + (t + \alpha )^2 f^2 - \tx\frac{1}{2b} \lf(2 f^2 - f^4 \ri) \ri\},
\eeq
and the first correction by $-\eps k$ times 
\begin{equation}\label{eq:corr func}
\fc _\alpha [f] := \int_{0} ^{c_0 |\log \eps|} \diff t \:  t\lf\{ \lf| \partial_t f \ri|^2 + f ^2 \left( -\alpha (t+\alpha) -\frac{1}{b} + \frac{1}{2b} f ^2\right)\ri\}
\end{equation}
which is obtained by retaining only linear terms in $\eps k$ when expanding~\eqref{eq:1D func}. We shall denote $\eoneo$ the minimum of $\fone_{0,\alpha}[f]$, and $\alO,\fO$ a minimizing pair ($\alO$ is unique, $\fO$ also is, up to a sign). Note the mild abuse of notation: the $k=0$ functional is well-defined even for $c_0 = +\infty$, and we take this convention. Due to known decay estimates for $\fO$ (see, e.g., \cite[Proposition 3.3]{CR1}), this creates only an exponentially small discrepancy in $\eps$, provided $c_0$ is a large enough constant.

We previously proved energy estimates relating the full GL energy to the infimum of the above 1D, curvature-dependent functional. Since our method was local, and the GL energy density is related to $|\glm| ^4$, it is natural to expect a result about the distribution of the latter quantity from the energy estimate. The goal of this paper is to provide this estimate.

Let us first introduce scaled boundary coordinates: the surface superconductivity layer 
\beq
	\label{eq:intro ann}
	 \annt : = \lf\{ \rv \in \Omega \: | \: \tau \leq c_0 \eps |\log\eps| \ri\},
\eeq
where
\beq
	\tau : = \dist(\rv, \partial \Omega),
\eeq
can be mapped to 
\begin{equation}\label{eq:intro def ann rescale}
\ann:= \left\{ (s,t) \in \left[0, |\partial \Omega| \right] \times \left[0,c_0 |\log\eps|\right] \right\}
\end{equation}
via a diffeomorphism $\Phi$.

For technical reasons (see Remark~\ref{rem:measures} below), we can only evaluate $\int_{D} |\glm| ^4$ with the desired precision in the case that the set $D$ looks ``rectangular'' in boundary coordinates. Let then $D\subset \Omega$ be a measurable set {\it independent of $ \eps $} such that
\begin{equation}\label{eq:rectangular}
\Phi( D \cap \annt ) = [s_D,s'_D] \times [0,c_0 |\log \eps|]  
\end{equation}
for some $s_D,s'_D \in [0, |\dd \Omega|] $. Notice that this implies that the boundary of $ D $ intersects $\dd \Om$ with $\pi/2$ angles. Our main result is the following:

\begin{teo}[\textbf{Curvature dependence of the order parameter}]\label{theo:main}\mbox{}\\
Let $\glm$ be a GL minimizer and $D\subset \Omega$ be a measurable set such that~\eqref{eq:rectangular} holds. Denote $s\mapsto k(s)$ the curvature of $\dd \Om$ as a (smooth) function of the tangential coordinate $s$.  For any $1<b<\theo ^{-1}$, in the limit $\eps \to 0$,
\begin{equation}\label{eq:main estimate}
\int_{D} \diff \rv \: |\glm| ^4 = \eps \, C_1(b) |\dd \Om \cap \dd D| + \eps ^2 C_2 (b) \int_{\dd D\cap \dd \Om} \diff s \: k(s) + o(\eps ^2),
\end{equation}
where $\diff s$ stands for the 1D Lebesgue measure along $\dd \Om$ and
\begin{align}\label{eq:lead ord}
C_1(b) &= - 2 b \eoneo = \int_{0} ^{+\infty} \diff t \: \fO ^4 > 0 
\\
\label{eq:curv corr} C_2 (b) &=  2 b \, \fc_{\alO} [f_0] = \tx\frac{2}{3}b \fO ^2 (0) - 2b \alO \eoneo.
\end{align}
Moreover, for $|b-\theo ^{-1}|$ small enough (independently of $\eps$), $C_2 (b) > 0.$
\end{teo}

\noindent
The leading order term in~\eqref{eq:main estimate} had previously been computed~\cite{FK,Kac,Pan}, with a less explicit expression of the constant $C_1 (b)$ however. 

\begin{remark}(Concentration of Cooper pairs.)\mbox{}\\
For obvious physical reasons it would be preferable to have an estimate of $ \int_D |\glm| ^2$, which would directly give information on the distribution of Cooper pairs close to the boundary of the sample. Unfortunately, our method, which is mostly energy-based, does not provide this. An estimate on $ \int_D |\glm| ^4 $ is easier to obtain because more directly linked to the concentration of the energy density 
\begin{equation}\label{eq:ener dens}
\gled(\rv) := \bigg| \bigg( \nabla + i \frac{\aavm}{\eps^2} \bigg) \glm \bigg|^2 - \frac{1}{2 \hex \eps^2} \lf( 2|\glm|^2 - |\glm|^4 \ri) + \frac{\hex}{\eps^4} \lf| \curl \aavm - 1 \ri|^2, 
\end{equation}
as we shall see below. The above theorem still indicates that, to leading order, $|\glm|$ is concentrated evenly along the boundary of $\Om$, and that the first correction to this effect is directly proportional to the curvature function $k(s)$. Notice that, in order to exploit \eqref{eq:ener dens}, the use of the variational equation solved by the GL minimizer is crucial and therefore our result does not extend directly to low energy configurations, although one would expect it.
\end{remark}

\begin{remark}(Convergence as measures.)\label{rem:measures} \mbox{}	\\
It would be natural to reformulate the result in terms of convergence of measures, i.e., by stating that
\begin{equation}\label{eq:curv dep psi 2}
\frac{1}{\eps}\left( \frac{1}{\eps} |\glm| ^4 \diff \rv -   C_1(b) \diff s(\rv)\right) \underset{\eps \to 0}{\longrightarrow} C_2 (b)  k(s) \diff s(\rv), 
\end{equation}
in the sense of measure,  where $\diff s(\rv) $ is the 1D Lebesgue measure along the boundary of $\Om$. However, due to the restriction \eqref{eq:rectangular} on the shape of the set $ D $, the statement \eqref{eq:main estimate} is weaker than \eqref{eq:curv dep psi 2}.

We in fact expect~\eqref{eq:main estimate} to be wrong as stated if the set $ D $ intersects the boundary with an angle $\neq \pi/2 $. Indeed, the result is based on integrating the energy density on lines normal to the boundary that cover the full extent of the physical region. These integrals give the energy of the 1D model~\eqref{eq:1D func} that enters the main formula. One can easily decompose $ D \cap \ann $ into the union of sets of the form \eqref{eq:rectangular} and some remainder which is more triangular-like. In the latter regions we cannot follow the proof procedure to obtain a simple expression. This does not affect the leading order of the result, as noted in~\cite{FK,Pan}, since the energy contained in the triangular regions is small relatively to the leading order. It is however \emph{not} small compared to the correction we isolate in~\eqref{eq:main estimate} because the area of the triangular regions can easily be seen to be $ O(\eps^2) $. Integrating $ |\glm|^4 $ on such an area gives a 
$O(\eps ^2)$ 
contribution and thus ruins the result.
 
Technically speaking, Assumption~\eqref{eq:rectangular} enters the proof in the estimate of the boundary integral appearing in Lemma~\ref{lem:ord param 1}. It allows to exploit a new pointwise estimate of the tangential derivative of $ |\glm|^2 $ that we prove in Lemma~\ref{lem:est gradient s}. It is important to remark that such an estimate does not hold for the normal component of the gradient (see the discussion preceding Lemma \ref{lem:est gradient s}).
\label{rem:rectangular}
\end{remark}

\begin{remark}(Limiting regimes.)
\mbox{}	\\
In \cite[Corollary 1.3]{FK} a similar result about the boundary behavior of surface superconductivity is proven for magnetic fields slightly below the second critical one $ \Hcc $, i.e., for $ b \to 1^- $, $ b \leq 1 $. In this estimate the leading order is the same as in \eqref{eq:main estimate} but the first correction is proportional to the area of the set $ |D| $ and is due to the bulk behavior of the superconductor. 

	The regime where $ b \to \theo^{-1} $ at the same time as $\eps \to 0$ was studied in details in \cite{FH1} (see in particular \cite[Theorem 1.4]{FH1} and the discussion thereafter). The behavior of the GL functional becomes approximately linear in this limit, and therefore the whole machinery of spectral theory of linear operators can be exploited to extract a lot of details about the boundary behavior. In particular, according to the asymptotics of $ \Theta_0^{-1} - b $ when $\eps \to 0$, superconductivity can be either uniformly distributed all over the boundary or concentrated close to the points of maximal curvature. The first order correction to the boundary behavior is however not known, although thanks to the approximate linearity of the problem, estimates can be formulated in terms of the integral of $ |\glm|^2 $ instead of $ |\glm|^4 $.
\end{remark}

\begin{remark}(Sign of the curvature correction)
\mbox{}	\\
Unfortunately we are not able to determine the sign of the correction in \eqref{eq:main estimate}. Based on the results of~\cite{FH1} we have just recalled, we conjecture that $ \fc_{\alpha_0} [f_0] > 0 $ for any $ 1 < b  < \Theta_0^{-1} $. This would mean that points with large curvature attract more superconductivity. We can prove this conjecture only when $b$ is close (independently of $\eps$) to $\theo ^{-1}$, see Lemma~\ref{lem:sign}. 

Note that the 2D setting we consider here corresponds to an infinite 3D cylinder with a magnetic field parallel to the axis. In a more general 3D setting, the angle between the magnetic field and the surface of the sample also plays a role in the distribution of surface superconductivity, see~\cite{FKP} and references therein.
\end{remark} 

The rest of the paper contains the proof of Theorem \ref{theo:main}. In Section~\ref{sec:1D} we first discuss the perturbative expansion of the 1D ground state energy, obtain the expression~\eqref{eq:curv corr} and prove that it is positive close to the third critical field. Section~\ref{sec:ener dens} contains a result of the same form as Theorem~\ref{theo:main} where $|\glm| ^4$ is replaced by the GL energy density~\eqref{eq:ener dens}. Finally, in Section~\ref{sec:ord param}, which is the more involved of the three, we deduce our main result from the estimate of the energy density. 

\medskip

\noindent\textbf{Acknowledgments.} M.C. acknowledges the support of MIUR through the FIR grant 2013 ``Condensed Matter in Mathematical Physics (Cond-Math)'' (code RBFR13WAET). N.R. acknowledges the support of the ANR project Mathostaq (ANR-13-JS01-0005-01). We are indebted to one of the anonymous referees of our previous paper~\cite{CR2}, whose remarks motivated the present investigation.

\section{Correction to the 1D Energy}\label{sec:1D}


The asymptotic expansion of the 1D energy that is behind \eqref{eq:main estimate} is a direct consequence of first-order perturbation theory. We first prove that the functional \eqref{eq:corr func} gives the subleading order of the 1D energy. 

Let us recall the notation: for any $ k \in \R $, $\eone_\star(k) $ is the infimum of \eqref{eq:1D func} w.r.t. both $f $ and $ \alpha $, while $f_0$ and $\alpha_0$ stand for the minimizing pair of \eqref{eq:1D func bis}.

\begin{lem}[\textbf{Perturbative expansion of the 1D energy}]\label{lem:1D sublead}\mbox{}\\
As $\eps \to 0$  
\begin{align}\label{eq:1D pert}
\eone_\star \left(k\right) &= \eoneo - \eps k \fc_{\alpha_0} [f_0] + O (\eps ^{3/2} |\log \eps| ^{\gamma}) \nonumber\\
&= - \frac{1}{2b} \int_0 ^{+\infty} \fO ^4 - \eps k \fc_{\alpha_0} [f_0] + O (\eps ^{3/2} |\log \eps| ^{\gamma})
\end{align}
for some fixed $\gamma >0$ where $\fc$ is defined as in~\eqref{eq:corr func}.
\end{lem}

\begin{proof}
We first take $f_0,\alpha_0$ as a trial pair for the functional~\eqref{eq:1D func}, which yields
$$ \eone_\star \left(k\right) \leq \eoneo - \eps k \fc_{\alpha_0} [f_0] + O (\eps ^2) $$
by expanding $\fone_{k,\alpha_0}[f_0]$ and using that $\alpha_0,f_0$ do not depend on $\eps$. In order to estimate the remainders it suffices to use the exponential decay of $f_0 $ \cite[Proposition 3.3]{CR1}, which implies that
\bdm
	\int_0^{\infty} \diff t \: t^n f_0^2(t) = O(1),
\edm
for any $ n \in \N $. The decay estimate also implies that we make no significant error by considering the $k=0$ functional as being defined on the whole half-line, provided $c_0$ is large enough.

Next we write 
$$ \eone_\star \left(k\right) = \fone_{k,\alk} [\fk] = \fone_{0,\alk} [\fk] - \eps k \fc_{\alk} [f_k] + O (\eps ^2) \geq \eoneo - \eps k \fc_{\alk} [f_k] + O (\eps ^2),$$
where we use the variational principle defining $\eoneo$. Finally it easily follows from the estimates of~\cite[Proposition 1]{CR2} that 
$$ \fc_{\alk} [f_k] = \fc_{\alO} [f_0] + O (\eps ^{1/2} |\log \eps| ^{\gamma})$$ for some $\gamma >0$. Gathering the previous inequalities, we get an upper and a lower bound to $\eone_\star \left(k\right)$ which together give~\eqref{eq:1D pert}. That 
\begin{equation}\label{eq:1D ener f4} 
\eoneo = -\frac{1}{2b} \int_{0} ^{+\infty} \diff t \: \fO ^4  
\end{equation}
follows by mutliplying the variational equation 
\begin{equation}\label{eq:var eq fO}
 - \dd_t ^2 \fO + (t+ \alO) ^2 \fO  = \frac{1}{b} (1-\fO ^2) \fO 
\end{equation}
 by $ f_0 $ and integrating.
\end{proof}

The expression~\eqref{eq:curv corr} is somewhat unwieldy, but can be simplified a lot: 

\begin{lem}[\textbf{Simple formula for the 1D energy correction}]\label{lem:expression}\mbox{}\\
Let $\alO,f_0$ be a minimizing pair for the 1D functional~\eqref{eq:1D func bis} at $k=0$. Then, for all $1 < b < \theo ^{-1}$,
\begin{align}\label{eq:corr final exp}
\fc_{\alO} [f_0] &=  \frac{2}{3} (1-\alO ^2 b) +\frac{ \alO}{2b} \int_0 ^{+\infty} \diff t \: \fO ^4 \nonumber\\
&= \frac{1}{3} \fO ^2 (0) - \alO \eoneo.
\end{align}
\end{lem}

\begin{proof}
Let us first recall two useful identities:
\begin{equation}\label{eq:opt cond}
\int_{0} ^{+\infty} \diff t\: (t + \alO) f_0 ^2 (t) = 0 
\end{equation}
expresses the optimality of $\alO$, see, e.g., \cite[Eq. (3.20)]{FHP} or \cite[Lemma 3.1]{CR1}, while 
\begin{equation}\label{eq:virial}
\int_{0} ^{+\infty} \diff t\: t (t + \alO) f_0 ^2 (t) = \int_{0} ^{+\infty} \diff t\: \left( |\dd_t f_0 | ^2 + \frac{1}{4b} f_0 ^4  \right) 
\end{equation}
is a virial identity, equivalent to~\cite[Eq. (3.22)]{FHP}. It is obtained by noting that 
$$ E_{\alO} (\ell) := \inf_f \foneO \left[\frac{1}{\sqrt{\ell}} f ( \cdot /\ell)\right] = \inf_f \foneO [f] $$
for all $\ell$ and thus 
$$ \lf. \dd_\ell E_{\alO} (\ell) \ri|_{\ell = 1} = 0. $$
Using the Feynman-Hellmann principle to evaluate the latter expression one gets~\eqref{eq:virial}. 

We now start the computation:

\noindent\textbf{Step 1.} We claim that
\begin{equation}\label{eq:corr bis}
\fc_{\alO} [f_0] = \int_0 ^{+\infty} \diff t \: \left( \frac{t}{b} - \frac{3\alO}{4b}\right)\fO ^4 (t) + \int_0 ^{+\infty} \diff t \: \left(2 t ^3 -\frac{2t}{b} - 2 \alO ^2 t + \frac{\alO}{b}\right)\fO ^2 (t). 
\end{equation}
First, using the virial identity~\eqref{eq:virial} we obtain 
$$ \fc_{\alO} [f_0] = \int_0 ^{+\infty} \diff t \lf\{  (t-\alO) |\dd_t \fO| ^2 + \left( \frac{t}{2b} - \frac{\alO}{4b}\right)\fO ^4 (t) - \frac{t}{b} \fO ^2 \ri\}.$$
Next, integrating by parts and using the Neumann boundary condition for $\fO$, 
$$ \int_0 ^{+\infty} \diff t \: (t-\alO) |\dd_t \fO| ^2 = - \int_0 ^{+\infty} \diff t \: (t-\alO) ^2 \dd_t \fO \dd_t ^2 \fO.$$
Inserting the variational equation~\eqref{eq:var eq fO} and integrating by parts again we deduce 
\bmln{
 \int_0 ^{+\infty} \diff t \: (t-\alO) |\dd_t \fO| ^2  = \int_0 ^{+\infty}\diff t \: \bigg\{ (t-\alO) ^2 (t+\alO) + (t+\alO) ^2 (t-\alO) \\ 
 \lf. - \frac{t-\alO}{b} + \frac{t-\alO}{2b} \fO ^2 (t) \ri\}  \fO ^2 (t) + \frac{\alO ^4}{2} \fO ^2 (0) - \frac{\alO ^2}{2b} \fO ^2 (0) + \frac{\alO ^2}{4b} \fO ^4 (0). 
}
The claim then follows from the identity (see, e.g., \cite[Proof of Lemma 3.3]{CR1})
\begin{equation}\label{eq:f boundary}
 \fO ^2 (0) = 2 - 2 \alO ^2 b
\end{equation}
and a bit of algebra.

\medskip

\noindent\textbf{Step 2.} Next we compute that 
\begin{equation}\label{eq:third moment}
\int_{0} ^{+\infty} \diff t \: (t+\alO) ^3 \fO ^2 = - \frac{1}{2b} \int_0 ^{+\infty} \diff t \: (t+\alO) \fO ^4 + \frac{1}{3} \left( 1- \alO ^2 b \right), 
\end{equation}
following~\cite[Proof of Lemma~3.2.7]{FH-book}. Let us denote 
$$ H_{\alO} : = -\dd_t ^2 + (t+\alO) ^2 + \tx\frac{1}{b} \fO ^2 - \frac{1}{b}$$
and recall that $\fO$ is a ground state for this Schr\"odinger operator, with eigenvalue $0$. For any function $v$ we therefore have, integrating by parts 
$$ \left\langle \fO, H_{\alO} v \right\rangle = v' (0) \fO (0).$$
If we apply this with $v =   (t+\alO) ^2 \fO' -  (t + \alO) \fO$, we find 
$$ \left\langle \fO, H_{\alO} v \right\rangle = \left(  \alO ^4 +   \frac{\alO ^2}{b}(\fO ^2 (0)-1) -   1\right) \fO^2 (0).$$
On the other hand, using the variational equation, 
\begin{align*}
 -\dd_t ^2 v &= - 3 (t + \alO) \fO ^{\prime\prime} - (t + \alO)^2 \fO ^{\prime\prime\prime} \\
 &= - (t+\alO) \left( 5(t+\alO) ^2 + \tx\frac{3}{b}  (\fO ^2 - 1)\right)  \fO
 -  (t+\alO )^2 \left( (t+\alO) ^2 + \tx\frac{3}{b} \fO ^2 - \frac{1}{b} \right) \fO^ {\prime}, 
 \end{align*}
so that 
$$ H_{\alO} v = - 6 (t+\alO) ^3 \fO - \tx\frac{4}{b} (t+\alO) (\fO ^2 - 1) \fO - \frac{2}{b} (t+\alO) ^2 \fO ^2 \fO ^{\prime}.$$
Multiplying this by $\fO$ and integrating yields
$$ \left\langle \fO, H_{\alO} v \right\rangle = -6 \int_{0} ^{+\infty} \diff t \: (t+\alO) ^3 \fO ^2  - \frac{3}{b} \int_0 ^{+\infty} \diff t \: (t+\alO) \fO ^4  + \frac{\alO^2}{2b} \fO ^4 (0)$$
where we use also~\eqref{eq:opt cond} and an integration by parts. Equating the two different expressions of $\left\langle \fO, H_{\alO} v \right\rangle$ we obtained and noting that, since $\fO ^2 (0) = 2 - 2 \alO^2 b$, 
$$ \fO^2 (0) \left( -\alO ^4  - \frac{\alO ^2}{b}  (\fO ^2 (0)-1) + 1\right) + \frac{\alO^2}{2b} \fO ^4 (0) = \fO ^2 (0)  = 2\left( 1- \alO ^2 b\right) $$
yields~\eqref{eq:third moment}.
 
\medskip

\noindent\textbf{Step 3.} From~\eqref{eq:third moment},~\eqref{eq:opt cond} and~\eqref{eq:virial} we get 
\begin{align*}
 \int_{0} ^{+\infty} \diff t \: t ^3 \fO ^2&= \int_{0} ^{+\infty} \diff t \lf\{ (t+\alO) ^3  - 3 \alO ^2  t    - 3 \alO   t^2   - \alO ^3   \ri\} \fO ^2\\
 &= \frac{1}{6}\fO ^2 (0) + \int_{0} ^{+\infty} \diff t \lf\{ - \frac{1}{2b} (t+\alO) \fO ^4 - 3 \alO  |\dd_t \fO| ^2 - \frac{3\alO}{4b}  \fO ^4 - \alO ^3   \fO ^2 \ri\}.
\end{align*}
Inserting in~\eqref{eq:corr bis} and using~\eqref{eq:opt cond} again we find 
\begin{align}\label{eq:near}
 \fc_{\alO} [f_0] &= \frac{1}{3}\fO ^2 (0) + \int_{0} ^{+\infty} \diff t \lf\{- \frac{13\alO}{4b} \fO ^4  - 6 \alO  |\dd_t \fO| ^2 +  \left( \frac{\alO}{b} - \frac{2t}{b} - 2 \alO ^2 t - 2\alO ^3\right)\fO ^2  \ri\} \nonumber\\
 &= \frac{1}{3}\fO ^2 (0) + \int_{0} ^{+\infty} \diff t \lf\{ - \frac{13\alO}{4b}\fO ^4  - 6 \alO     |\dd_t \fO| ^2 + 3 \frac{\alO}{b}     \fO ^2 \ri\}.
\end{align}
Finally we note that
\begin{align*}
 \int_{0} ^{+\infty} \diff t \: |\dd_t \fO| ^2 &= \eoneo - \frac{1}{2b} \int_{0} ^{+\infty} \diff t \: \fO ^4 +  \int_{0} ^{+\infty} \diff t \lf( \frac{1}{b}  -  (t+\alO) ^2 \ri) \fO ^2 \\
 &= -\frac{1}{b}\int_{0} ^{+\infty} \diff t \: \fO ^4 +  \int_{0} ^{+\infty} \diff t \lf( \frac{1}{b}  -  (t+\alO) ^2 \ri) \fO ^2
\end{align*}
whereas~\eqref{eq:opt cond} and~\eqref{eq:virial} together imply
$$ \int_{0} ^{+\infty} \diff t \: (t+\alO) ^2 \fO ^2 = \int_{0} ^{+\infty} \diff t \lf\{ |\dd_t \fO| ^2 + \frac{1}{4b}  \fO ^4 \ri\} $$
so that, combining the two identities,
$$ \int_{0} ^{+\infty} \diff t \: |\dd_t \fO| ^2 =  \int_{0} ^{+\infty}\diff t \: \lf\{ \frac{1}{2b}\fO ^2 - \frac{5}{8b}  \fO ^4 \ri\}.$$
Inserting this in~\eqref{eq:near} and recalling once more~\eqref{eq:f boundary} and~\eqref{eq:1D ener f4} this yields the final expressions~\eqref{eq:corr final exp}.
\end{proof}

Unfortunately we are not able to determine the sign of the energy correction from the expressions we found. However, when $b\to \theo ^{-1}$ we have more information: it is known that $\fO ^2$ scales as $\left(1 - b\theo\right) ^{1/2}$. It then immediately follows from Lemma~\ref{lem:expression} that the correction must be positive for $b$ close enough to $\theo ^{-1}$:

\begin{lem}[\textbf{Sign of the correction close to $\Hccc$}]\label{lem:sign}\mbox{}\\
There exists $ 1< b_0 < \theo ^{-1}$ such that, for all $b_0 < b < \theo^{-1}$,   
\begin{equation}\label{eq:sign Hccc}
\fc_{\alO} [f_0] > 0. 
\end{equation}
\end{lem}

\begin{proof}
We recall from~\cite[Section~3.2]{FH-book} that the minimum in~\eqref{eq:def theo} is achieved by a unique pair $u_0,\alpha_{\mathrm{opt}}$ with $\alpha_{\mathrm{opt}} = - \sqrt{\theo}$ and $ u_0 $ normalized in $ L^2 $. At $b>\theo ^{-1}$ it is easy to see that $f_0\equiv 0$. When $b\to \theo ^{-1}$ one should therefore expect $\fO \to 0$, in which case the quartic term becomes a second order correction. One should thus expect that the solution is close to that of the linear problem, the only effect of the quartic term being to fix the overall normalization. More precisely, following the techniques of~\cite[Section~14.2.2]{FH-book} one can show that 
$$ \alO \underset{b\to \theo ^{-1}}{\longrightarrow} - \sqrt{\theo}$$
and that 
$$ \left(\frac{\lf\| u_0 \ri\|_4^4 }{ 1- b \theo}\right) ^{1/2} \fO \underset{b\to \theo ^{-1}}{\longrightarrow} u_0.$$
It is easy to see that the latter convergence holds in the quadratic form domain of the harmonic oscillator. Using standard elliptic estimates, one can upgrade this to any  Sobolev or H\"older norm. In particular, convergence holds in $L ^{\infty}$ and $L^4$, so that 
$$ \fO ^2 (0) = \frac{(1 -b\theo)}{\lf\| u_0 \ri\|_4^4}  u_0 ^2 (0) (1+o_b(1))$$
where we denoted by $ o_b(1) $ a quantity going to $ 0 $ as $ b \to \theo^{-1} $, and 
$$ \int_{0} ^{+\infty} \diff t \: \fO ^4 = \frac{(1 -b\theo) ^2}{\lf\| u_0 \ri\|_4^4} (1+ o_b(1)).$$
From Lemma~\ref{lem:expression} we have that  
\bmln{
	\fc_{\alO} [f_0] = \frac{1}{3} \fO ^2 (0) + \frac{ \alO}{2b} \int_0 ^{+\infty} \diff t \: \fO ^4 = \frac{1 - b \theo}{3\lf\| u_0 \ri\|_4^4}  \lf[ u_0^2(0) + \frac{3 \alO}{2b}(1 - b \theo) \ri] 	\\
	= \frac{1 - b \theo}{3\lf\| u_0 \ri\|_4^4} \lf[ u_0^2(0) + o_b(1) \ri] > 0
}
since $ u_0(0) > 0 $ is independent of $ b $.
\end{proof}

\section{Estimates of the Energy Density}\label{sec:ener dens}

Our main estimate on the order parameter is obtained by exploiting the link between $|\glm| ^4$ and the GL energy density. We discuss first the asymptotics for the latter.

\begin{pro}[\textbf{Estimates for the energy density}]\label{pro:ener dens}\mbox{}\\
Let $\gled$ be the GL energy density defined in~\eqref{eq:ener dens} and $D \subset \Om$ be a measurable set satisfying~\eqref{eq:rectangular}.
Under the same assumptions and with the same notation as in Theorem~\ref{theo:main} we have, as $\eps \to 0$
\begin{equation}\label{eq:ener dens estimate}
\int_{D} \diff \rv \: \gled = \eps ^{-1} \eoneo |\dd \Om \cap \dd D| -  \fc_{\alO} [f_0] \int_{\dd D\cap \dd \Om} \diff s \: k(s) + O(\eps ^{1/2} |\log \eps| ^{\gamma}).  
\end{equation}
\end{pro}

\begin{proof}
This is an adaptation of the method developed in~\cite{CR1,CR2}. First we recall from~\cite[Lemma 4]{CR2} the energy lower bound
\beq
			\label{eq:energy lb ann}
			\glee \geq \frac{1}{\eps} \annf[\psi] -C \eps^2 |\log\eps|^2,
		\eeq
where $\psi$ is, up to a phase factor, the GL order parameter in boundary coordinates and the reduced functional is 
\begin{multline}\label{eq:GL func bound} 
	\annf[\psi] : = \int_0^{|\partial \Omega|} \diff s \int_0^{c_0 |\log\eps|} \diff t \lf(1 - \eps \curv t \ri) \lf\{ \lf| \partial_t \psi \ri|^2 + \frac{1}{(1- \eps \curv t)^2} \lf| \lf( \eps \partial_s + i \aae(s,t) \ri) \psi \ri|^2 \ri.	\\	
	\lf. - \frac{1}{2 \hex} \lf[ 2|\psi|^2 - |\psi|^4 \ri]  \ri\}
\end{multline}
with 
\beq\label{eq:vect pot bound}
	\aae(s,t) : =- t + \half \eps \curv t^2 + \eps \deps , 	
\eeq
and
\beq\label{eq:deps}
	\deps : = \frac{\gamma_0}{\eps^2} - \lf\lfloor \frac{\gamma_0}{\eps^2} \ri\rfloor,	\qquad	 \gamma_0 : = \frac{1}{|\partial \Omega|} \int_{\Omega} \diff \rv \: \curl \, \aavm,
\eeq
$ \lf\lfloor \: \cdot \: \ri\rfloor $ standing for the integer part.

This lower bound may in fact (with identical proof) be localized, yielding 
\begin{multline}\label{eq:low bound loc}
\eps \int_{D} \diff \rv \: \gled \geq \int_{s_D}^{s'_D} \diff s \int_0^{c_0 |\log\eps|} \diff t \lf(1 - \eps \curv t \ri) \lf\{ \lf| \partial_t \psi \ri|^2 \ri.	\\
	\lf. + \frac{1}{(1- \eps \curv t)^2} \lf| \lf( \eps \partial_s + i \aae(s,t) \ri) \psi \ri|^2 - \frac{1}{2 \hex} \lf[ 2|\psi|^2 - |\psi|^4 \ri]  \ri\} -C \eps^3|\log\eps|^2. 
\end{multline}
We next split the interval $[s_D,s'_D]$ into $N_\eps = O (\eps ^{-1})$ sub-intervals $[s_n,s_{n+1}]$, $n=1, \ldots, N_\eps$, of side length $O(\eps)$. The convention is that $s_D = s_1$ and $s'_D = s_{N_\eps +1}$. This gives a decomposition of $[s_D,s'_D] \times [0,c_0 |\log \eps|]$ into $N_\eps$ rectangular cells $C_n$, $n=1, \ldots, N_\eps$.

Arguing as in~\cite[Lemma 6]{CR2} we then deduce
\begin{equation}\label{eq:reduc func}
	 \eps \int_{D} \diff \rv \: \gled  \geq \int_{s_D}^{s'_D} \diff s \: \eones(k(s)) + \sum_{n=1} ^{\neps} \E_n[u_n] - C\eps^2 \logi
	 \end{equation}
where, within the $n$-th cell, 
\beq
			\label{eq:splitting psi}
			\psi(s,t) = : u_n(s,t) f_n(t) \exp \lf\{-i \tx\left(\frac{\alj}{\eps} + \deps\right)s \ri\},
		\eeq
		and the reduced functionals $\E_n$ are defined as 
		\bml{
			\label{eq:Ej}
			\E_n [u] : =  \int_{\cellj} \diff s \diff t \lf(1 - \eps \kj t \ri) f_n^2 \lf\{ \lf| \partial_t u \ri|^2 + \tx\frac{1}{(1- \eps \kj t)^2} \lf| \eps \partial_s u \ri|^2 - 2 \eps b_n(t) J_s[u] \ri.	\\
			\lf. + \tx\frac{1}{2 \hex} f_n^2 \lf(1 - |u|^2 \ri)^2  \ri\},	
		}
		with
		\beq
			b_n(t) : = \frac{t + \alj - \half \eps \kj t^2}{(1 - \eps \kj t)^2},
		\eeq
		and
		\beq
			J_s[u] : = (i u, \partial_s u) = \tx\frac{i}{2} \lf(u^* \partial_s u - u \partial_s u^*\ri).
		\eeq
		The mean curvature in the $n$-th cell is denoted $k_n$, with $f_n$ and $\alj$ the minimizing profile and phase for the associated functional~\eqref{eq:1D func}.
Inserting the result of Lemma~\ref{lem:1D sublead} into~\eqref{eq:reduc func} we find 
\bml{\label{eq:reduc bis}
	  \int_{D} \diff \rv \: \gled  \geq \eps ^{-1} \eoneo |\dd \Om \cap \dd D| -  \fc_{\alO} [f_0] \int_{\dd D\cap \dd \Om} \diff s \: k(s) 	\\
	  + \eps ^{-1} \sum_{n=1} ^{\neps} \E_n[u_n] - C\eps \logi.
}

We next adapt the strategy of~\cite[Proof of Lemma 7]{CR2} to estimate the reduced functionals from below: 
\begin{equation}\label{eq:low bound reduc loc}
  \sum_{n=1} ^{\neps} \E_n[u_n] \geq - C \eps ^{3/2} |\log \eps| ^{\gamma}, 
\end{equation}
thereby concluding the proof of the lower bound corresponding to~\eqref{eq:ener dens estimate}. This is a long procedure that we will not recall in details. We shall emphasize the only point that has to be modified, due to the fact that we now bound from below the energy density in the set $D$ instead of the full GL energy. 

Step 1 of the proof of~\cite[Lemma 7]{CR2} adapts with no modification, leading to
\bml{
 			\label{eq:step 4}
			\E_n[u_n] \geq  \eps \int_0^{\tj} \diff t F_n(t) \lf[ J_t[u_n](s_{n+1}, t) -  J_t[u_n](s_{n}, t) \ri] \\
			 + \de \int_{\cellj} \diff s \diff t \: \lf(1 - \eps \kj t \ri) f_n^2 \lf[ \lf| \partial_t u_n \ri|^2 
			+ \tx\frac{1}{(1- \eps \kj t)^2} \lf| \eps \partial_s u_n \ri|^2 \ri]  + O (\eps ^{\infty}),
		}
for some $\de \sim |\log \eps | ^{-4}$, denoting
$$
J_t[u] : = (i u, \partial_t u) = \tx\frac{i}{2} \lf(u^* \partial_t u - u \partial_t u^*\ri)
$$
and
$$ F_n (t) = 2 \int_0 ^t \diff \eta\: (1 - \eps k_n \eta) b_n (\eta) f_n ^2 (\eta).$$
We then follow Step 2 of the same proof to combine and estimate the boundary terms produced by the use of Stokes' formula in Step 1 (terms on the first line of the above formula). In this procedure it is crucial to sum boundary terms living on the same cell boundary. Since here $s_1 \neq s_{N_\eps +1}$ there is obviously a need for a different estimate of the terms located on the corresponding boundaries, i.e., those of the original set $D$. This is the only point where we depart slightly from the method of~\cite{CR2} and rely on more refined inequalities. 

We proceed as follows (say for the $n=1$ term, located on the boundary corresponding to $s_1 = s_D$): let $\chi$ be a smooth cut-off function depending only on $s$ with 
$$ \chi (s_1) = 1,\quad  |\chi| \leq 1, \quad \supp (\chi) \subset \mathcal{C}_1, \quad |\nabla \chi| \leq C \eps ^{-1}. $$
Since our cells have side-length $O(\eps)$ in the $s$ direction, the last two requirements are obviously compatible. Intergrating by parts in the $s$ variable we get
\begin{align}\label{eq:split bound term}
 \int_0^{\tone} \diff t \: F_1(t)  J_t[u_1](s_{1}, t)  &= \int_0^{\tone} \diff t \: \chi (s_1) F_1(t)  J_t[u_1](s_{1}, t)  \nonumber\\
 &= \int_{\cell_1} \diff s \diff t\: \chi F_1 \dd_s J_t [u_1] + \int_{\cell_1} \diff s \diff t\: F_1 \dd_s \chi J_t [u_1]. 
\end{align}
We drop the subscripts $1$ for shortness. To handle the first term in the above we note that 
$$ \dd_s J_t [u] = \frac{i}{2}\left( \dd_s u \dd_t u ^* - \dd_s u^* \dd_t u \right) + \frac{i}{2}\left( u\dd_s\dd_t u ^* - u ^* \dd_s \dd_t u \right)$$
and hence a further integration by parts in $t$ yields 
$$ \int_{\cell} \diff s \diff t\: \chi F \dd_s J_t [u] = - \frac{i}{2}\int_{\cell} \diff s \diff t\: \chi \dd_t F  (u\dd_s u^* - u^* \dd_s u).$$
Note that the boundary terms vanish by definition of $F$. We also have $|\dd_t F| \leq C |\log \eps| f ^2$ and thus 
\begin{align*}
 \eps \left| \int_{\cell}\diff s \diff t\: \chi F \dd_s J_t [u]\right| &\leq C \eps |\log \eps| \int_{\cell}\diff s \diff t\: f ^2 |u| |\dd_s u|\\
 &\leq C \delta |\log \eps| \int_{\cell}\diff s \diff t\: f ^2 |u| ^2 + C \delta ^{-1} |\log \eps| \int_{\cell}\diff s \diff t\: f ^2 |\eps \dd_s u| ^2\\
 &\leq C \delta \eps |\log \eps|^2  + C \delta^{-1} \eps^2 |\log \eps|^{\gamma} \leq C \eps^{3/2} |\log \eps|^\gamma
\end{align*}
where we use that $f^2 |u|^2 = |\psi| ^2 \leq 1 $ plus the fact that $|\cell| = O (\eps |\log \eps|)$, recall the estimate~\cite[Eq. (6.15)]{CR2} and have chosen $\delta = \eps ^{1/2}|\log \eps|^\gamma$ for the final optimization.

For the second term in~\eqref{eq:split bound term} we write, using essentially the same ingredients (in particular~\cite[Eq. (6.15)]{CR2}),
\begin{align*}
\left|\eps \int_{\cell} \diff s \diff t\:F \dd_s \chi J_t [u] \right| &\leq C \int_{\cell} \diff s \diff t\:f ^2 |u| |\dd_t u |\\
&\leq C\delta \int_{\cell} \diff s \diff t\:f ^2 |u| ^2 + C \delta ^{-1} \int_{\cell} \diff s \diff t\:f ^2 |\dd_t u | ^2 \\
&\leq C \delta \eps |\log \eps| + C \delta ^{-1} \eps ^{2} | \log \eps| ^{\gamma} \leq C \eps ^{3/2} |\log \eps| ^\gamma. 
\end{align*}
Combining the previous estimates, we obtain 
$$ \eps \int_0^{\tone} \diff t \: F(t)  J_t[u](s_{1}, t) = O(\eps ^{3/2} |\log \eps| ^{\gamma})$$
and a similar estimate for the term located on the boundary $s=s_{N_{\eps} + 1}$. Dealing with the other boundary terms as in~\cite{CR2} concludes the proof of~\eqref{eq:low bound reduc loc}. 

At this stage we have the lower bound corresponding to~\eqref{eq:ener dens estimate},
\begin{equation}\label{eq:low bound loc pre}
 \int_{D} \diff \rv \: \gled \geq \eps ^{-1} \eoneo |\dd \Om \cap \dd D| -  \fc_{\alO} [f_0] \int_{\dd D\cap \dd \Om} \diff s \: k(s) + O(\eps |\log \eps| ^{\gamma})   
\end{equation}
and by the same method also
\begin{equation}\label{eq:low bound comp}
 \int_{D ^c} \diff \rv \: \gled \geq \eps ^{-1} \eoneo |\dd \Om \cap \dd D ^c| -  \fc_{\alO} [f_0] \int_{\dd D ^c \cap \dd \Om} \diff s \: k(s) + O(\eps ^{1/2} |\log \eps| ^{\gamma}).  
\end{equation}
On the other hand, combining the energy estimate of~\cite[Theorem 1]{CR2} and Lemma~\ref{lem:1D sublead} we have the global estimate
$$ \int_\Om \diff \rv \: \gled = \eps ^{-1} \eoneo |\dd \Om| -  \fc_{\alO} [f_0] \int_{\dd \Om} \diff s \: k(s) + O(\eps |\log \eps| ^{\gamma}).$$
Combining with~\eqref{eq:low bound comp} we deduce
\bmln{
\int_{D} \diff \rv \: \gled =  \int_\Om \diff \rv \: \gled - \int_{D^c} \diff \rv \: \gled	\\ \leq \eps ^{-1} \eoneo |\dd \Om \cap \dd D| -  \fc_{\alO} [f_0] \int_{\dd D\cap \dd \Om} \diff s \: k(s) + O(\eps ^{1/2} |\log \eps| ^{\gamma}) 
}
which we combine with~\eqref{eq:low bound loc pre} to complete the proof.
\end{proof}

\section{From the Energy Density to the Order Parameter}\label{sec:ord param}

We now conclude the proof of Theorem~\ref{theo:main} by adding the following to Proposition~\ref{pro:ener dens}:

\begin{pro}[\textbf{Energy density versus order parameter}]\label{pro:ord param}\mbox{}\\
Under the assumptions and with the notation of Theorem~\ref{theo:main} and Proposition~\ref{pro:ener dens}, we have 
\begin{equation}\label{eq:ord param}
\int_D \diff \rv \: \gled = - \frac{1}{2\hex\eps ^2} \int_D \diff\rv \: |\glm| ^4 +  o (1). 
\end{equation}
\end{pro}

The proof is split in two lemmas. First we have a general result which does not require the set under consideration to be rectangular:

	\begin{lem}[\textbf{Reduction to a boundary term}]\label{lem:ord param 1}\mbox{}\\
		Let $S\subset \Om$ be a measurable set. Then, with the previous notation
		\begin{equation}\label{eq:ord param 1}
		\int_S \diff \rv \: \gled  + \frac{1}{2\hex\eps ^2} \int_S \diff \rv \: |\glm| ^4  = \frac{1}{2}\int_{\dd S } \diff \sigma \: \nabla |\glm| ^2 \cdot \nuv +  o (1) 
		\end{equation}
		with $\nuv$ the outward-pointing normal to $\dd S$.
	\end{lem}

\begin{proof}
We first note that 
\begin{equation}\label{eq:neglect mag}
 \frac{1}{\eps^4} \int_S \diff \rv \: \lf| \curl \aavm - 1 \ri|^2 = O(\eps |\log \eps| ^3).  
\end{equation}
Indeed, using the elliptic estimate (see~\cite[Eq. (11.51)]{FH-book}) 
$$ \left\Vert \curl \aavm - 1 \right\Vert_{C^1 (\Om)} = O(\eps) $$
and the fact that $\curl \aavm = 1$ on $\dd \Om$, we deduce that in the full boundary layer~\eqref{eq:intro ann} we have 
$$ \left| \curl \aavm - 1 \right| = O( \eps ^2 |\log \eps|)$$
and thus 
$$ \int_{S\cap \annt} \diff \rv \: \lf| \curl \aavm - 1 \ri|^2 = O(\eps ^5 |\log \eps| ^3).$$
The part of the integral located in $S \cap \annt ^c$ is of much lower order, as follows from the usual Agmon estimates, for instance~\cite[Eq. (12.10)]{FH-book}, and we deduce~\eqref{eq:neglect mag}.  

At the level of precision we aim at we may thus neglect the magnetic kinetic energy:
$$ \int_S \diff \rv \: \gled = \int_{S} \diff \rv \: \bigg\{ \bigg| \bigg( \nabla + i \frac{\aavm}{\eps^2} \bigg) \glm \bigg|^2 - \frac{1}{2 \hex \eps^2} \lf( 2|\glm|^2 - |\glm|^4 \ri) \bigg\} + O (\eps |\log \eps| ^3).$$
Next we recall that since $\glm$ is a minimizer for $\glfe$ we have the first Ginzburg-Landau variational equation:
\beq
	\label{eq:GL var eq}
	 -\left( \nabla + i \frac{\aavm}{\eps ^2} \right) ^2 \glm + \frac{1}{\hex \eps ^2} \glm \left( |\glm| ^2 - 1 \right) = 0.
\eeq
Combining with the identity
$$ \tx\frac{1}{2} \Delta |\glm|^2 = \Re (\overline{\glm} \Delta \glm) + |\nabla \glm| ^2$$
we deduce 
\beq 
	\label{eq:Gl var ineq}
	\frac{1}{2}\Delta |\glm| ^2 = \bigg| \bigg( \nabla + i \frac{\aavm}{\eps^2} \bigg) \glm \bigg|^2 + \frac{1}{\hex \eps ^2} |\glm| ^2 \left(|\glm| ^2 -1 \right).
\eeq
Integrating over $S$ we obtain 
\begin{equation}\label{eq:ener ord par}
 \int_S \diff \rv \: \gled = - \frac{1}{2\hex \eps ^2} \int_S |\glm | ^4 + \frac{1}{2} \int_{S} \diff \rv \: \Delta |\glm| ^2 
\end{equation}
and the proof is complete since of course
$$ \int_{S} \diff \rv \: \Delta |\glm| ^2 = \int_{\dd S } \diff \sigma \: \nabla |\glm| ^2 \cdot \nuv.$$
\end{proof}

Applying the previous lemma with $D = S$, our task should now be to estimate the boundary term in the right-hand side of~\eqref{eq:ord param 1}. It is similar to terms showing up in~\cite[Proof of Lemma 6.1]{FK}, but using the estimates therein shows at best that it is of order $O(1)$, a remainder that we cannot afford. A technical novelty in the present paper is thus to show that this term is in fact $o(1)$, provided the set $D$ is rectangular in boundary coordinates. 

We certainly have 
$$ \int_{\dd D } \diff \sigma \: \nabla |\glm| ^2 \cdot \nuv = \int_{(\dd D )\cap \ann} \diff \sigma \: \nabla |\glm| ^2 \cdot \nuv + O (\eps ^{\infty})$$
by the usual decay estimates, where we recall that $\ann$ is defined in~\eqref{eq:intro def ann rescale}. Splitting the curve $\LL_D := \Phi (\dd D \cap \ann)$ (which is a rectangle in boundary coordinates) in a part $\LL_D^s$ parallel to the boundary of $\dd \Om$ and a part $\LL_D^t$ normal to the boundary of $\dd \Om$ we have
\beq
	\label{eq:ord param 2} \int_{\dd D } \diff \sigma \: \nabla |\glm| ^2 \cdot \nuv = \int_{\LL_D ^t} \diff \sigma \: \nabla |\glm| ^2 \cdot \nuv + O(\eps ^{\infty}).
\eeq
Indeed, on the part of $\LL_D ^s$ which coincides with $\dd \Om$ we have 
$$\nabla |\glm| \cdot \nuv = 0$$
by taking the real part of the Neumann boundary condition
$$ \left(\nabla \glm + i \frac{\aavm}{\eps^2} \right)\cdot \nuv = 0,	\qquad		\mbox{on } \dd \Om$$
satisfied by $\glm$. The other part of $\LL_D ^s$ is deep in the region where the order parameter decays exponentially and may thus be neglected. The new key ingredient is that on $\LL_D ^t$ we can prove 
\begin{equation}\label{eq:grad s moral}
 \lf| \nabla |\glm| ^2 \cdot \nuv \ri| = \lf| \dd_s |\glm| ^2 \ri| \leq C \eps ^{a-1} 
\end{equation}
for some $a>0$. This is natural in view of the results of~\cite{CR1,CR2}: the variations in the $s$-direction should be much smaller than those in the $t$-direction, which happen on a scale~$\eps$. Combining this with the fact that the length of this part of the boundary is $O(\eps |\log \eps|)$ we get the desired estimate. 

We are in fact not able to prove~\eqref{eq:grad s moral} in all the boundary region $ \ann $ and will thus have to split the line integral on $ \LL^t_D $ into two pieces. Let us introduce the following subset of $ \ann $ where the estimate \eqref{eq:grad s moral} is proven in the next Lemma \ref{lem:est gradient s}:
\beq
	\label{eq:abt}
	\abt : = \lf\{ \rv \in \Omega \: \big| \: f_0(\tau/\eps) \geq \eps^{1/6} \ri\}.
\eeq
Any power of $ \eps $ strictly smaller than $ 1/4 $ would do the job but we fix it equal to $ 1/6 $ for concreteness. Recall that $ \tau = \dist(\rv, \partial \Omega) $. As before we denote by $ \ab  = \Phi(\abt) $, the set $ \abt $ in boundary coordinates and it is easy to see that
\beq
	\label{eq:ab}
	\ab = \lf[0,\partial \Omega\ri] \times \lf[0, t_>\ri],	\qquad		t_> \gg 1.
\eeq
In fact, exploiting the available pointwise bounds on $ \fO $ (see for example~\cite[Eq. (A.6)]{CR2}), one immediately verifies that
\beq
	\label{eq:tlarger}
	t_> \geq \tx\frac{1}{\sqrt{3}} \sqrt{|\log\eps|} (1 + o(1)).
\eeq
Before stating the pointwise estimate mentioned above, let us stress that a similar bound cannot hold for the normal component of the gradient of $ \glm $: the estimate $ \dd_t |\glm| ^2 \propto \eps ^{-1} $ is optimal. The different behavior of the $s$ and $t$ derivatives will be apparent in the proof of the following Lemma. 

	\begin{lem}[\textbf{Estimate of the tangential derivative}]
		\label{lem:est gradient s}
		\mbox{}	\\
		As $ \eps \to 0 $
		\beq
			\label{eq:est gradient s}
			\lf\| \partial_s \lf| \glm(\Phi(s,\tau)) \ri|^2 \ri\|_{L^{\infty}(\ab)} = O(\eps^{-5/6} |\log\eps|^\infty).
		\eeq
	\end{lem}
	
	\begin{proof}
		The starting point is the variational equation \eqref{eq:GL var eq} satisfied by $ \glm $: setting as in \cite[Sect. 5]{CR2} 
		\beq
			\label{eq:psi}
			\psi(s,t) = \glm(\Phi(s,\eps t)) e^{-i \phi_{\eps}(s,t)},
		\eeq
		where the explicit expression of the gauge phase is given in \cite[Eq. (5.4)]{CR2}, one gets
		\beq
			- \partial_t^2 \psi + \tx\frac{1}{(1 - \eps k(s) t)^2} \lf(-i \eps \partial_s + \tx\frac{\tilde{A}}{\eps} \ri)^2 \psi = \frac{1}{b} \lf(1 - \lf| \psi \ri|^2 \ri) \psi,
		\eeq
		i.e., thanks to the choice of the gauge the magnetic field is now purely tangential. The explicit expression of $ \tilde{A} $ can be easily recovered in terms of $ \phi_{\eps} $ (see, e.g., \cite[Eq. (5.6)]{CR2}), but the most important point is the estimate 
		\bdm
			\big\| \tilde{A}(s,t) + \eps t \big\|_{L^{\infty}(\ann)} = O(\eps^2|\log\eps|^2),
		\edm
		which follows from a priori estimates on $ \aavm $ as \cite[Eq. (4.23)]{CR1}. In addition the explicit expression \cite[Eq. (5.6)]{CR2} also implies that $ \big\| \partial_s \tilde{A} \big\|_{\infty} = O(\eps|\log\eps|) $. Notice that here we are exploiting the smoothness of $ \partial \Omega $ and the fact that the curvature is infinitely differentiable. Plugging the ansatz 
		\beq
			\label{eq:u}
			\psi(s,t) = f_0(t) e^{-i \frac{\alpha_0}{\eps} s} u(s,t),
		\eeq
		for some unknown function $ u $ and with $ f_0 $ and $ \alpha_0 $ the minimizing  density and phase of the half-plane functional~\eqref{eq:1D func bis}, we get
		\bdm
			- \partial_t^2 \lf( f_0 u \ri) + \tx\frac{f_0}{(1 - \eps k(s) t)^2} \lf( -i \eps \partial_s + \alpha_0 + t + O(\eps|\log\eps|) \ri)^2 u = \frac{1}{b} \lf(1 -  f_0^2 |u|^2 \ri) f_0 u.
		\edm
		Exploiting now the variational equation~\eqref{eq:var eq fO} for $ f_0 $ and dividing by $f_0 > 0 $, we obtain
		\bdm
			- \partial_t^2 u - 2 \tx\frac{f_0^{\prime}}{f_0} \partial_t u - \tx\frac{1}{(1 - \eps k(s) t)} \eps^2 \partial_s^2 	u - 2 i \eps (\alpha_0 + t) \partial_s u = \frac{f_0^2}{b} \lf(1 -  |u|^2 + O(\eps|\log\eps|^2) \ri) u. 
		\edm
		Since (see \cite[Lemma A.1]{CR1})
		\bdm
			\frac{f_0^{\prime}}{f_0} = O(|\log\eps|^3),	
		\edm
		inside $ \ab $ the above equation yields the estimate
		\bmln{
			\lf| \left(\partial_t ^2 + \eps ^2 \partial_s ^2 \right) u \ri| \leq C \lf[ |\log\eps|^3 \lf| \left(\partial_t + \eps\partial_s\right) u \ri| + \lf| 1 - |u| \ri| \lf|u\ri| \ri] \\
			\leq C \lf[ |\log\eps|^3 \lf| (\partial_t,\eps\partial_s) u \ri| + \eps^{1/12} |\log\eps|^{b} \ri] 
		}
		where we have used the upper bound $ |u| \leq f_0^{-1} \leq \eps^{-1/6} $ and the estimate 
		$$ \lf| 1 - |u| \ri| = O(\eps^{1/4} |\log\eps|^{\infty}) \mbox{ in } \ab $$
		which follows from~\cite[Theorem 2]{CR2}. Rescaling now also the $ s $ variable by setting $ \xi = s/\eps $ and denoting $ v(\xi,t) = u(\eps \xi ,t) $, we can apply the Gagliardo-Nirenberg inequality~\cite[p. 125]{N} 
		\bdm
			\lf\| \nabla_{\xi,t} v \ri\|_{\infty} \leq C \lf( \lf\| \Delta_{\xi,t} v \ri\|^{1/2}_{\infty} \lf\| 1 - |v| \ri\|^{1/2}_{\infty} + \lf\| 1 - |v| \ri\|_{\infty} \ri),
		\edm
		which implies $ \lf\| \nabla v \ri\|_{\infty} = O(\eps^{1/6} |\log\eps|^{\infty}) $ and therefore
		\beq
			\lf\| \lf( \partial_t + \eps \partial_s \ri) |u| \ri\|_{L^{\infty}(\ab)} \leq \lf\| \nabla v \ri\|_{\infty} = O(\eps^{1/6} |\log\eps|^{\infty}).
		\eeq 	
		The result on $ \glm $ then follows from the identities~\eqref{eq:psi} and~\eqref{eq:u}.
	\end{proof}

Putting together the results of Lemma~\ref{lem:ord param 1} and Lemma~\ref{lem:est gradient s}, we are now in position to complete the proof of the main result of this Section:

	\begin{proof}[Proof of Proposition \ref{pro:ord param}]
		The estimate \eqref{eq:ord param 1} and the properties of the set $ D $ yield \eqref{eq:ord param 2} and therefore
		\bml{
			\label{eq:boundary int 1}
			\int_S \diff \rv \: \gled  + \frac{1}{2\hex\eps ^2} \int_S \diff \rv \: |\glm| ^4  = \frac{1}{2}\int_{\dd S } \diff \sigma \: \nabla |\glm| ^2 \cdot \nuv +  o (1) \\
			= 	\int_{\LL_D ^t} \diff \sigma \: \nabla |\glm| ^2 \cdot \nuv +  o (1) 
			= \int_{\LL_D ^t \cap \abt} \diff \sigma \: \nabla |\glm| ^2 \cdot \nuv +  o (1),
		}
		since 
		\bmln{
			\int_{\LL_D ^t \cap \abt^c} \diff \sigma \: \nabla |\glm| ^2 \cdot \nuv = 2 \int_{\Phi(\LL_D ^t \cap \abt^c)} \diff \tau \: |\psi(s,\tau/\eps)| \partial_s |\psi(s,\tau/\eps)| \\
			= 2 \int_{\Phi(\LL_D ^t \cap \abt^c)} \diff \tau \: |\psi(s,\tau/\eps)| f_0(\tau/\eps) \partial_s |u(s,\tau/\eps)| \leq C \eps^{-1}  \lf\| \psi \ri\|_{L^{\infty}(\ab^c)} \lf| \Phi(\LL_D ^t \cap \abt^c) \ri|		
		}
		where we have used \cite[Eq. (6.2)]{CR2}. Now the Agmon estimate for $ \psi $ stated, e.g., in \cite[Eq. (5.5)]{CR2} yields
		\bdm
			 \lf\| \psi \ri\|_{L^{\infty}(\ab^c)} \leq C e^{ - A t_> } \leq C \exp\lf\{ - \tx\frac{1}{\sqrt{3}} \sqrt{|\log\eps|} \ri\} \ll |\log\eps|^{-1},
		\edm
		thanks to \eqref{eq:tlarger}. Hence we conclude that
		\bdm
			\int_{\LL_D ^t \cap \abt^c} \diff \sigma \: \nabla |\glm| ^2 \cdot \nuv \leq C |\log\eps|  \lf\| \psi \ri\|_{L^{\infty}(\ab^c)} = o(1),
		\edm
		and \eqref{eq:boundary int 1} is proven. For the rest of the boundary integral it suffices to apply \eqref{eq:est gradient s}:
		\bdm
			\int_{\LL_D ^t \cap \abt} \diff \sigma \: \nabla |\glm| ^2 \cdot \nuv = \int_{\Phi(\LL_D ^t) \cap \ab} \diff \sigma \: \partial_s \lf|\glm(\Phi(s,\tau))\ri|^2 = O(\eps^{1/6} \logi).
		\edm
	\end{proof}

\end{document}